\documentclass[12pt]{elsarticle}

\usepackage{xcolor}
\usepackage{times}
\usepackage{natbib}
\usepackage{array,xspace,multirow,hhline,tikz,colortbl,tabularx,booktabs,fixltx2e,amsmath,amssymb,amsfonts,amsthm}
\usepackage{subfig}
\usepackage{float}
\usepackage{natbib}
\usepackage{booktabs}
\usepackage{tikz}
\usepackage{pgflibraryshapes}  
\usetikzlibrary{positioning}

\usepackage{setspace}
\usepackage{threeparttable} 

\usepackage{mathtools}
\usepackage{amsmath}

\newcommand{\eg}{e.g.,\xspace}

\newtheorem{lemma}{Lemma}%
\newtheorem{theorem}{Theorem}%
\newtheorem{corollary}{Corollary}%
\newtheorem{example}{Example}

\DeclareMathOperator{\closure}{\textrm{closure}}
\DeclareMathOperator{\dom}{\textrm{dom}}
\DeclareMathOperator{\increment}{\textrm{next}}

\newcommand{\strpref}{P}
\newcommand{\sympref}{I}

\newcommand{\rsd}[0]{\ensuremath{\mathit{RSD}}\xspace}

\newcommand{\set}[1]{\left\{ #1 \right\}}
\newcommand{\sset}[2]{\left\{ #1 : #2 \right\}}
\newcommand{\abs}[1]{\left| #1 \right|}

\newcommand{\pref}{R\xspace}

\usepackage[textsize=tiny,textwidth=3cm]{todonotes}

\usepackage{lineno}

\begin{document}

\begin{frontmatter}

	\title{Parametrized Algorithms for \\Random Serial Dictatorship}

		\author[nicta]{Haris Aziz\corref{cor1}}  \ead{haris.aziz@nicta.com.au}		
	\address[nicta]{NICTA and UNSW, 223 Anzac Parade, Sydney, NSW 2033, Australia, \\Phone: +61 2 8306 0490}

		\author[syd]{Juli\'{a}n Mestre} \ead{mestre@it.usyd.edu.au}
		\address[syd]{School of Information Technologies, The University of Sydney, Australia, \\Phone: +61 2 9351 4276}

\begin{abstract}
	
	Voting and assignment are two of the most fundamental settings in social
	choice theory. For both settings, \emph{random serial dictatorship} (\rsd) is
	a well-known rule that satisfies anonymity, ex post efficiency, and
	strategyproofness. Recently, it was shown that computing the resulting
	probabilities is \#P-complete both in the voting and assignment setting. In
	this paper, we present efficient parametrized algorithms to compute the \rsd
	probabilities for  parameters such as the number of agent types,
	alternatives, or objects. When the parameters are small, then the respective algorithms are considerably more efficient than the naive approach of going through all permutations of agents. 

\end{abstract}


\end{frontmatter}


\section{Introduction}

Voting and assignment are two of the most widely applied and important settings
in social choice theory. 
In both settings, although one could also
use \emph{discrete} or \emph{deterministic} rules, randomization is crucial to achieve minimal fairness
requirements such as anonymity and neutrality.
In voting, \emph{agents} express preferences over
\emph{alternatives} and a \emph{social decision scheme} returns a probability
distribution over the alternatives based on the agents'
preferences~\citep{Barb79b,Gibb77a,Proc10a}. In the assignment setting,
\emph{agents} express preferences over \emph{objects} and a \emph{random
assignment rule} returns a random assignment of the objects specifying the
probability with which each object is allocated to each agent~\citep{BCK11a,
BoMo01a, BCKM12a}. 
The objects are referred to as \emph{houses} in the literature and the assignment setting is also known as \emph{house allocation}. 
For the two settings, \rsd is a desirable social decision
scheme~\citep{ABBH12a, Gibb77a} and random assignment rule~\citep{BoMo01a,CrMo01a}, respectively.

In the voting setting, \emph{random serial dictatorship} (\rsd) takes a permutation of agents uniformly at random
and then selects an alternative by serially allowing agents in the permutation
to refine the set of feasible alternatives. In the assignment setting, \rsd
takes a permutation uniformly at random and then lets the agents in the
permutation serially take their most preferred house that has not yet been
allocated.

For both the settings mentioned above, \rsd is a well-known rule that is
anonymous, strategyproof, and ex post efficient (randomizes over Pareto optimal
alternatives). In fact, it has been conjectured to be the only rule that
satisfies these properties~\citep[see \eg][]{LeSe11a,PaSe13a}. \rsd is well-established
and commonly used especially in resource allocation. In particular, the
resulting probabilities of \rsd can be viewed as fractional allocations in
scheduling and other applications (in which the houses are in fact divisible) and hence important to compute~\citep[see
\eg][]{AbSo98a,BoMo01a,BuCa12a,CrMo01a,Sven94a}. 
Similarly, in voting, the
probabilities returned by the \rsd rule can be interpreted as fractions of
time or resource allotted to the alternatives and hence crucial to compute. The probability of each alternative can also be used as a suggestion for the proportional representation of the alternative in representative democracy or seat allocation of a parliament~\citep{Aziz13b,Tull67a}.

The definition of \rsd for both the settings suggests natural exponential-time algorithms to compute the \rsd probabilities: enumerate all the permutations and for each permutation, perform a linear number of operations. However these algorithms are naive and the question arises whether there are more efficient algorithms to compute the \rsd probabilities. 

Recently, \citet{ABB13b} showed that the resulting probabilities of \rsd are
\#P-complete to compute both in the voting and the assignment settings.
Independently, \citet{SaSe13a} also showed the same result for the assignment
problem. In view of the inherent complexity of computing the \rsd
probabilities, \citet{SaSe13a} mentioned that identifying the conditions under
which the problem is polynomial time as an open problem. \citet{MeSe13a}
propose random hybrid assignment mechanisms that hinge on the $\rsd$
probabilities. They termed the problem of computing \rsd probabilities in the
assignment domain as a ``very difficult research problem.''

In view of the importance of \rsd in both voting and resource allocation and
the recent negative computational results, we undertake an algorithmic study of
\rsd. We show that \rsd is amenable to efficient computation provided certain
structural parameters are small.
More
precisely, we undertake a \emph{parametrized} complexity analysis of \rsd for
both voting and assignment.


\paragraph{Standard vs.\@ parametrized complexity}

In \emph{standard computational complexity theory}, only the size of the
problem instance is considered as a measure of the problem's complexity. An
algorithm is deemed to be efficient if its running time is bounded by a
polynomial function of the size of the instance; in other words, for every
instance $I$ of the problem, the running time of the algorithm is bounded by
$poly(|I|)$ where $poly$ is a fixed polynomial independent of
the instance and $|I|$ is the size of the instance. Unfortunately, not every
problem is known to admit an efficient algorithm. Indeed, researchers have
identified certain classes of problems that are thought not to admit efficient
algorithms. For example, the fact that $\rsd$ is \#P-complete
\cite{ABB13b,SaSe13a} strongly suggests that there is no efficient algorithm
for computing the $\rsd$ probabilities.

In \emph{parametrized complexity theory}, a finer multivariate analysis is
undertaken by considering multiple \emph{parameters} of the problem
instance~\citep{DoFe13a,Nied06a}. Intuively, a parameter is some aspect of the problem input. For example, for computational problems on graphs, the maximum degree of the graph is a natural parameter.  
Let $k$ be a parameter of an
instance $I$. A problem with parameter $k$ belongs to the class FPT, or is said
to be \emph{fixed-parameter tractable}, if there exists an algorithm that
solves the problem in $f(k) \cdot poly(|I|)$ time, where~$f$ is some computable
function and $poly$ is a polynomial both independent of $I$. The key idea behind FPT is to
separate out the complexity into two components---a component $poly(|I|)$ that
depends solely on the size of the input, and a component $f(k)$ that depends on
the parameter. An FPT algorithm with parameter $k$ can solve instances in which
the input size of the instance is large as long as $k$ is small and hence the
growth of $f(k)$ is relatively small. Of course if $k$ is not small and $f(k)$
is, say, exponential in $k$, then the running time of the FPT algorithm can become
too slow for practical purposes. Nevertheless, designing FPT algorithms with different parameters is important as they expand
tractability frontier for problems that are otherwise intractable in general.


\paragraph{Contributions}

We propose algorithms to compute the \rsd probabilities that under reasonable
assumptions are significantly more efficient than the naive method of going
over $n!$ permutations of the agents. To be precise, we present FPT algorithms
with parameters such as \# agent types, \# alternatives, \# alternative types
and \# houses. Many of our algorithms exploit different dynamic programming
formulations where the recursion is based on new insights into \rsd applied to
voting and assignment. In this sense, our work not only yields more efficient
algorithms for computing an \rsd lottery, but also furthers our understanding
of a keystone algorithm in social choice theory.

\section{Voting Setting}

We first define the voting setting and \rsd formally. We follow the
notation used in \citep{ABBH12a}. A voting setting consists of a set
$N=\{1,\ldots, n\}$ of \emph{agents} having preferences over a finite set $A$
of \emph{alternatives} where $|A|=m$. The preferences of agents over alternatives are
represented by a preference profile $\pref=(\pref_1,\ldots, \pref_n)$ where
each agent $i\in N$ has complete and transitive preferences $\pref_i$ over $A$.
By $(a,b) \in \pref_i$, also denoted by $a \mathrel{\pref_i} b$, we mean that
alternative $a$ is at least as preferred by agent $i$ as alternative $b$. We
denote with $\strpref_i$ the strict part of~$\pref_i$ ($a \mathrel{\strpref_i} b$
if~$a \mathrel{\pref_i} b$ but not~$b\mathrel{\pref_i} a$), and with~$\sympref_i$ the
symmetric part of~$\pref_i$ ($a \mathrel{\sympref_i} b$ if~$a \mathrel{\pref_i}
b$ and~$b\mathrel{\pref_i} a$). A preference relation $\pref_i$ is
\emph{linear} if $a \mathrel{\strpref_i} b$ or~$b\mathrel{\strpref_i} a$ for
all distinct alternatives $a,b \in A$.
The size of an instance of a voting setting will be denoted by $|\pref|$ which is equal to $|N|\times |A|$. 

We let $\Pi^N$ denote the
set of all permutations of $N$ and write a permutation $\pi \in \Pi^N$ as
$\pi =  \pi(1) \ldots \pi(n)$. If $R_i$ is a preference relation and
$B\subseteq A$ a subset of alternatives, then $\max_{R_i}(B)=\{ a\in B \colon
a \mathrel{R_i} b \text{ for all }b\in B\}$ is the set of most preferred
alternatives from $B$ according to $R_i$. Let $\pi(i)$ be the $i$-th agent in
permutation $\pi\in \Pi^N$. Then,
\begin{equation*}
	\textit{RSD}(N,A,\pref)=\sum_{\pi\in \Pi^N}
	\frac{1}{n!}\delta_{\text{uniform}}(\mathrm{Prio}(N,A,\pref,\pi))
\end{equation*}
where
\begin{equation*}
	\mathrm{Prio}(N,A,\pref,\pi)=\max_{\pref_{\pi(n)}}(\max_{\pref_{\pi(n-1)}}(\cdots
	(\max_{\pref_{\pi(1)}}(A))\cdots)),
\end{equation*}
and $\delta_{\text{uniform}}(B)$ is the uniform lottery over the multi-set
$B$. In the literature, $\mathrm{Prio}(N,A,\pref,\pi)$ is simply referred to as \emph{serial dictatorship} with respect to ordering $\pi$.
We illustrate how \rsd works with the aid of a simple example.

\begin{example}[Illustration of \rsd in voting]
	Consider the following preference profile.

	\begin{align*}
		1:&\quad a \mathrel{I_1} b \mathrel{I_1} c \mathrel{P_1} d\\
		2:&\quad b \mathrel{I_2} d  \mathrel{P_2} a  \mathrel{P_2} c\\
		3:&\quad  c  \mathrel{P_3} a \mathrel{I_3} b \mathrel{I_3} d 
	\end{align*}

Then let us consider the Prio outcomes for each permutation over the voters.

\begin{align*}
	123:&\quad\{b\}&132:&\quad\{c\}\\
	213:&\quad \{b\}&231:& \quad\{b\}\\
	312:&\quad\{c\}& 321:&\quad \{c\}		 
\end{align*}

Thus the \rsd lottery is $[a:0, b:1/2, c:1/2]$.
\end{example}

If each agent has a unique most preferred alternative, the $\rsd$ lottery can be computed in linear time: the \rsd probability of an alternative is the fraction of agents who express it as most preferred. However, the problem of computing \rsd probabilities become \#P-complete when agents do not express strict preferences.
As seen by the formal definition as well as the example above, \rsd
probabilities can be computed by enumerating all the $n!$ permutations over
agents. 
We present alternative algorithms to compute the \rsd probabilities
that are significantly faster provided certain structural parameters are
small.

Two agents are said to be of the same \textit{type} if they have identical
preferences. Two alternatives are of the same \textit{type} if every agent is
indifferent between them.  
The following lemma shows that without loss of
generality we can focus on instance where no two alternatives have the same
type.

\begin{lemma}\label{lemma:alternative-types}
	There is a linear-time reduction from general instances of {\sc RSD} to
	simplified instances of {\sc RSD} where no two alternatives have the same
	type.
\end{lemma}

\begin{proof}
	Given an instance $(N,A,\pref)$, we construct an equivalent \emph{simplified}
	instance $(N,A',\pref')$ by contracting all alternatives of the same type
	into a `super' alternative.

	Given an \rsd lottery for $(N, A', \pref')$ we can construct a lottery for
	$(N,A, \pref)$ by uniformly dividing the \rsd probability of the `super'
	alternative among the alternatives in $A$ that induced it. 
\end{proof}

Unless otherwise stated, from now on, we assume that we are dealing with 
\emph{simplified} instances where no two alternatives have the same type.

Let $a \in A$ be a fixed but arbitrary alternative. For each $i\in N$, we
define the signature of $i$ (with respect to $a$) to be $(C,D)$ where $C$ is the subset of alternatives that are as good as $a$, and $D$ is the subset of alternatives that are strictly better than $a$; more formally, $C =
\sset{b}{b \mathrel{\sympref_i} a}$ and $D = \sset{b}{b \mathrel{\strpref_i}
a}$. Notice that even if two agents have different types, they can still have the same signature. On the other hand, if two agents have the same
type, they must have the same signature.

For notational convenience we enumerate all the signatures and denote the set
of signatures with $\mathcal{S} = \set{ 1,2, \ldots }$. Let the $i$th signature
be defined by the pair $(C_i, D_i)$ and let $t_i > 0$ be the number of agents
having this signature. For a subset $X \subseteq \mathcal{S}$ of signatures we
use $t(X)$ to denote $\sum_{i \in X} t_i$. Since we are dealing with a
simplified instance, at least one agent is not indifferent between a given pair of alternatives and hence each Prio outcome results in a singleton set. Hence, it follows that $\cap_{i \in N} C_i = \set{a}$.
In this way, running the serial dictatorship on any permutation of the agents
either does not select $a$ or selects exactly $a$.

We show an FPT algorithm for computing $RSD(R)(a)$ where the parameter
is the number of signatures $\abs{\mathcal{S}}$.

\begin{theorem}  \label{thm:fpt}
  In the voting setting, for each alternative $a \in A$ there is an
  algorithm for computing $RSD(R)(a)$ that runs
  in $O(|R| + m |\mathcal{S}| \cdot 2^{|\mathcal{S}|})$ time.
\end{theorem}

\begin{proof}
	For each $X \subset \mathcal{S}$ we define a residual problem where the set
	of alternatives is $\bigcap_{i \in \mathcal{S} \setminus X} C_i$, where there
	are no agents with a signature in $\mathcal{S} \setminus X$, and where we
	still have $t_i$ agents with signature $i$ for each $i \in X$. For $X =
	\mathcal{S}$ the residual problem is the same as the original problem. In the
	residual problem defined by $X$ we say a permutation of its agents is
	\emph{lucky} if running the serial dictatorship with this permutation selects
	alternative $a$. Our algorithm is based on a dynamic programming formulation for counting the number of lucky permutations:
	\begin{equation}
	    M[X] =  \text{\# lucky permutations in the residual problem defined by $X$. }
	\end{equation}
	Notice that $RSD(R)(a) = M[\mathcal{S}]/ n!$, so if we can compute $M$, we are done.
	In the rest of the proof, we derive a recurrence to do just that.

	For each $X \subset \mathcal{S}$ we define the set of \emph{admissible signatures} to be
	\[\phi(X) = \sset{i \in X}{D_i \cap \bigcap_{j \in \mathcal{S} \setminus X} C_j=\emptyset}\]
	and
	\[ \phi(\mathcal{S}) = \sset{i \in \mathcal{S}}{ D_i = \emptyset}.\]
	The key observation is that every lucky permutation in the residual problem
	defined by $X$ must start with an agent having an admissible signature.

	To compute
	the value of the DP (dynamic program) states, we use the following recurrence.

	\begin{equation}
	    \label{eq:recurrence}
	    M[X] =  
	    \begin{cases}
	        t(X)! & \text{if } \phi(X) = X \\
	        0 & \text{if } \phi(X) = \emptyset \\
	        \displaystyle \sum_{i \in \phi(X)} t_i! { t(X) -1 \choose t_i - 1} M[X \setminus \set{i} ] & \text{otherwise}
	    \end{cases}
	\end{equation}

	Let us briefly justify each case of the recurrence. First, consider the case
	$\phi(X) = X$, which means that in the residual problem defined by $X$ every
	agent has $a$ in its top equivalence class. If that is the case, then every
	permutation of the agents is lucky. Since there are $t(X)$ agents in the
	residual problem, it follows that there are $t(X)!$ lucky permutations.

	Second, consider the case $\phi(X) = \emptyset$, which means that in the
	residual problem defined by $X$, not a single agent has $a$ in its top
	equivalence class. If that is the case, then there are no lucky permutations.

	Finally, consider the case $\emptyset \subset \phi(X) \subset X$. Recall that
	every lucky permutation must begin with an agent with a signature $i \in
	\phi(X)$. Notice that after such an agent is chosen the set
	of possible alternatives is reduced to $\bigcap_{j \in \mathcal{S} \setminus
	(X \cup \set{i})} C_j$. Also, the remaining agents $(t_i - 1)$ agents with
	signature  $i$ do not further constrain the set of alternatives. Therefore, if
	we take a lucky permutation in the residual problem defined by $X$ and we
	strip from it all agents with signature $i$, we are left with a lucky
	permutation for the residual problem defined by $X \cup \set{i}$. Similarly,
	if we take a lucky permutation for the residual problem defined by $X \cup
	\set{i}$ and we prepend one agent with signature $i$ and insert the remaining
	$t_i-1$ agents with signature $i$ any way we want, we have a lucky
	permutation for the residual problem defined by $X$. Notice that for each
	lucky permutation for $X \cup \set{i}$ there are ${ t(X) - 1
	\choose t_i -1 }$ ways of choosing the positions for the agents with
	signature $i$ and for each one of those there are $t_i!$ ways of distributing
	the individual agents. If follows that there are $\sum_{i \in \phi(X)} t_i!
	{t(X) - 1	\choose t_i -1} M[ X \setminus \set{i}]$ lucky permutations for the residual
	problems defined by $X$.

	Computing the signatures $\mathcal{S}$ and their frequencies $t_1, \ldots,
	t_{\abs{\mathcal{S}}}$ takes $O(\abs{R})$ time, 
	where $|R|$ is the total
	length of the agent preferences. 	Computing the admissible signature function
	$\phi$ can be done in $O(m \abs{\mathcal{S}} \cdot 2^{\abs{\mathcal{S}}})$
	time, where $m$ is the number of alternatives. The size of the DP table is
	$2^{|\mathcal{S}|}$, and given all the previous information, computing each
	entry of the table takes $O(\abs{\mathcal{S}})$ time. Hence, the total time
	to compute $RSD(R)(a)$ is $O(|R| + m |\mathcal{S}| \cdot 2^{|\mathcal{S}|}
	)$.
\end{proof}

\begin{corollary}\label{cor:numberofagents}
  There is an FPT algorithm for computing $RSD(R)(a)$ with parameter $n =
  \text{\# of agents}$. The running time is $O(|R| + m n \cdot 2^n)$.       \end{corollary}

\begin{proof}
	Two agents of the same type must have the same signature, so $|\mathcal{S}| \leq n$. 
	It follows from Theorem~\ref{thm:fpt} that the running time is $O(|R| + m n \cdot 2^n)$ $=O(|R|\cdot (1+2^n))=O(poly(|R|)\cdot (1+2^n))$ and hence the algorithm is FPT with parameter $n$.
\end{proof}

\begin{corollary}\label{cor:numberofagenttypes}
	There is an FPT algorithm for computing $RSD(R)(a)$	with parameter $T = \text{\# of agent types}$. The running time is $O(|R|+ m T \cdot 2^T)$.
\end{corollary}

\begin{proof}
In the worst case, each agent type has a different signature, so $|\mathcal{S}| \leq T$. It follows from Theorem~\ref{thm:fpt} that the running time is $O(|R|+ m T \cdot 2^T)=O(|R|\cdot (1+T \cdot 2^T))=O(poly(|R|)\cdot (1+T \cdot 2^T))$ and hence the algorithm is FPT with parameter $T$.
\end{proof}

\begin{corollary}\label{cor:alternatives}
    There is an FPT algorithm for computing $RSD(R)(a)$ with parameter $m =
    \text{\# of alternatives}$. The running time is $O(|R| + m \cdot 3^m \cdot 2^{3^m})$.
\end{corollary}
\begin{proof}
	This follows from Theorem~\ref{thm:fpt} and the fact that if there are $m$
	alternatives, there are at most $3^m$ different signatures $(C_i, D_i)$
	because
	\begin{equation*}
		\sum_{k = 1}^m {m \choose k} 2^{m-k} = 3^m,
	\end{equation*}
	where the first term of the left-hand-side product is the number of way of
	choosing a subset $C_i$ of size $k$ and second term is how many choices we
	have for $D_i$ given that $|C_i|=k$ and $C_i \cap D_i = \emptyset$.
	Therefore, $\abs{\mathcal{S}} \leq 3^m$.
\end{proof}


\begin{corollary}\label{cor:alternativestypes}
  There is an FPT algorithm for computing $RSD(R)(a)$ with parameter $q =
  \text{\# of alternative types}$. The running time is $O(|R| + q \cdot 3^q \cdot 2^{3^q})$.
\end{corollary}

\begin{proof}
	This follows from Lemma~\ref{lemma:alternative-types} and
	Corollary~\ref{cor:alternatives}.
\end{proof}

\section{Assignment Setting}
\label{sec:domain}

An \emph{assignment setting} is a triple $(N, H, \pref)$, where $N$ is a set of $n$ agents, $H$ is a set of $m\leq n$ houses, and
$\pref=(\pref_1,\ldots, \pref_n)$ is a preference profile that contains, for
each agent $i$, a \emph{linear} preference relation on some subset of houses that are \emph{acceptable} to the agent. If a house is \emph{unacceptable} to an agent, then he would prefer not to get any house than being allocated an unacceptable house. 
The size of an instance of an assignment setting will be denoted by $|\pref|$ which is equal to $|N|\times |H|$. 
Every randomized assignment yields a \emph{fractional assignment} that
specifies, for every agent $i$ and every house $h$, the probability $p_{ih}$
that house~$h$ is assigned to agent $i$. The fractional assignment can be seen
as a compact representation of the randomized assignment.  \rsd takes a
permutation uniformly at random and then lets the agents in the permutation
serially take their most preferred house that has not yet been allocated. If no acceptable houses remain, then the agent takes no house. 

It is easily observed that the assignment setting is a special case of a social
choice problem where $A$, the set of alternatives is the set of all discrete
assignments and the preferences of agents over $A$ are induced by their
preferences over $H$. Although agents have strict preferences
over the houses, they are indifferent among all assignments in which they
are allocated the same house. We illustrate how \rsd works as a random
assignment rule.

\begin{example}[Illustration of \rsd for random assignment]
	Consider the following preference profile.
		\begin{align*}
		1:&\quad a \mathrel{P_1} b \mathrel{P_1} c \\
		2:&\quad a \mathrel{P_2} b  \mathrel{P_2} c\\ 
		3:&\quad  b  \mathrel{P_3} a \mathrel{P_3} c
	\end{align*}

	Agent $1$ and $2$ are of the same type since they have the same preferences.
	Let us consider the Prio outcomes for each permutation over the voters. Each
	outcome is a discrete matching.
	\begin{align*}
		123: & \quad \{\{1,a\},\{2,b\},\{3,c\}\} &
	  132: & \quad \{\{1,a\},\{3,b\},\{2,c\}\} \\
		213: & \quad \{\{2,a\},\{1,b\},\{3,c\}\} &
		231: & \quad \{\{2,a\},\{3,b\},\{1,c\}\} \\
		312: & \quad \{\{3,b\},\{1,a\},\{2,c\}\} &
		321: & \quad \{\{3,b\},\{2,a\},\{1,c\}\}	 
	\end{align*}

	Thus the \rsd fractional assignment is obtained by taking the uniform convex
	combination of the discrete assignments for each of the permutations (see
	Table~\ref{table:randomassigmentRSD}).

	\begin{table}[h!]
		\centering	
		\begin{tabular}{l|lll}

	  \toprule
		&$a$&$b$&$c$\\ \midrule
		$1$&$1/2$&$1/6$&$1/3$\\
		$2$&$1/2$&$1/6$&$1/3$\\
		$3$&$0$&$2/3$&$1/3$\\
		\bottomrule

		\end{tabular}
		\caption{Fractional/randomized assignment as a result of \rsd.}
		\label{table:randomassigmentRSD}
	\end{table}

\end{example}

Although the assignment setting is a subdomain of social choice (voting), it does not mean that positive algorithmic results for voting imply the same for the assignment domain. The reason is that the transformation from an assignment setting to a corresponding voting problem leads to an exponential blowup in the number of alternatives.
Hence, results in the previous section 
do not carry over directly to the domain of assignments.

\citet{SaSe13a} showed that it is not possible (under suitable complexity-theoretic assumptions) to design an efficient algorithm for approximating the \rsd probabilities  even if randomization is used.
		Since neither randomization nor approximation is helpful, this further motivates a parametrized algorithm approach. The main result in this section is an FPT algorithm with composite parameter
number of houses and number of agent types. From this result we derive two
corollaries: There is a polynomial-time algorithm when the number of agent
types is constant, and an FPT algorithm with parameter number of houses.


We use $T$ to denote the number of agent types, and $d_j$ to denote the number
of agents of type $j$ in the instance. With a slight abuse of notation, we will
use $\pref_j$ to denote preference order of an agent of type $j$.

\begin{theorem}\label{th:fptassign}
	In the assignment setting, for each $i \in N$ and $h \in H$ there is an algorithm for computing
	$RSD(R)(i)(h)$ running in $O(|R| + { m+T \choose T } \cdot T \cdot 
	m^{T+1})$~time.
\end{theorem}	

\begin{proof}
	Let $\vec{s} = (s_1, \ldots, s_T)$ be a vector of integers where $0 \leq s_j
	\leq d_j$ for all $j=1, \ldots, T$. 
	Also let $\vec{b}=(b_1, \ldots, b_T)$ where $b_j \in H \cup
	\set{\mathrm{nil}}$ for $j = 1, \ldots, T$. Multiple entries of $\vec{b}$ can have the same house.
	Let $\dom(\vec{b})$ be those
	houses that are \emph{dominated} by the allocation $\vec{b}$, namely, for
	each house $h' \in \dom(\vec{b})$, there must be an agent type $j$ that
	prefers $h'$ to $b_j$; more formally,
	\begin{align*}
		\dom(\vec{b}) & =
		\sset{h' \in H}{ \exists \, j :
			\begin{aligned}
				& b_j \neq \textrm{nil} \quad \wedge \quad h' \mathrel{\strpref_j} b_j, \text{ or} \\
				& b_j = \textrm{nil} \quad \wedge \quad h' \text{ is acceptable to type } j
			\end{aligned}
		  }.
	\end{align*}


	For each $(\vec{s}, \vec{b})$ we define a residual problem where the set of
	houses available is $H \setminus \dom(\vec{b})$ and there are $s_j$ agents of
	type $j$ for each $j = 1, \ldots, T$. 
	 Intuitively, $\vec{b}$ is maintained in a way so that each $b_j$ is the house most preferred by agent type $j$ that has yet not been allocated. 
	Let $j^*$ be the type of agent $i$ from
	the theorem statement. Consider a permutation of the agents where there are
	$s_j$ agents of type $j$ and $i$ is one of the agents of type $s_{j^*}$---the
	precise identity of the other agents is not important. We say such a
	permutation is  \emph{lucky} if running serial dictatorship with this
	permutation on the residual instance results in assigning $i$ to $h$.

	We again use a dynamic programming formulation based on counting lucky
	permutations in the residual problems
	\begin{equation*}
		M[\vec{s}, \vec{b}] =
			\text{\# lucky permutations in the residual problem defined by $(\vec{s}, \vec{b})$}.
	\end{equation*}

	Our goal now is to derive a recurrence relation for $M[\vec{s}, \vec{b} ]$.
	We begin with some base cases. First, if there are no agents of type $j^*$
	left or if the house $h$ is not available anymore, then there are no lucky
	permutations; more formally,
	\begin{equation*}
		M[\vec{s}, \vec{b}] = 0 \quad \text{ if } s_{j^*} = 0 \text{ or } h \in \dom(\vec{b}).
	\end{equation*}
	
We say that $\vec{b}$ is \emph{degenerate} if $b_j \in \dom(\vec{b})$ for
	some $j$. For such degenerate vectors, let us define $\closure(\vec{b})$ to
	be the vector $\vec{\tilde{b}}$ such that $\tilde{b}_j$ is the most preferred
	house by agents of type $j$ that does not belong to $\dom(\vec{b})$ or
	$\textrm{nil}$ if all acceptable houses for agents of type $j$ belong to
	$\dom(\vec{b})$. It follows that if $\vec{b}$ is degenerate
	then
	\begin{equation*}
		M[ \vec{s},\vec{b}] = M[\vec{s},\closure(\vec{b})]. 
	\end{equation*}

	Let $\increment_j(h')$ be the house that comes after $h'$ in the total order
	of type $j$ preferences. If $h'$ happens to be the last acceptable house for
	type-$j$ agents then $\increment_j(h') = \textrm{nil}$. It is worth noting
	that $\closure(\vec{b})$ can be defined recursively in terms of the
	$\increment$ operator:
	\begin{equation*}
		\closure(\vec{b}) = 
		\begin{cases}
			\closure(\vec{b}_{-j}, \increment_j(b_j))
			& \text{if } \exists \, j : b_j \in \dom(\vec{b}), \\
			\vec{b} 
			& \text{if } \forall \, j : b_j \notin \dom(\vec{b}),
		\end{cases}
	\end{equation*}
	where the notation $(\vec{b}_{-j}, x)$ denotes the vector that is identical
	to $\vec{b}$ except for coordinate $j$, which takes the value $x$; in other
	words, 
	
	\[(\vec{b}_{-j}, x) = (b_1, \ldots, b_{j-1}, x, b_{j+1}, \ldots, b_T).\]
	
Our final corner case is that where there is an agent of type $j$ such that
	$b_j = \textrm{nil}$ and $s_j > 0$. In this case it does not matter where the
	remaining $s_j$ agents of type $j$ are placed in the permutation since none
	of them will get a house. Therefore, we get
	\[ M[\vec{s},\vec{b}] = { s_1 + \cdots + s_T \choose s_j} s_j! 
		\cdot M[(\vec{s}_{-j}, 0), \vec{b}]. \]

	For the recursive case of the recurrence, we condition on the type of 
	agent that is chosen to lead the lucky permutation. If the agent is of type
	$j \neq j^*$ and $b_j \neq \textrm{nil}$, assuming $s_j > 0$, there are $s_j$
	agents to choose from, so the number of such permutations will be \[ s_j
	\cdot M[(\vec{s}_{-j}, s_j -1), (\vec{b}_{-j}, \increment_j(b_j))]. \]
	If the agent is of type $j^*$, assuming $s_{j^*} > 0$, and the agent is not
	$i$, then the number of such permutations is \[ (s_{j^*} - 1) \cdot
	M[(\vec{s}_{-j^*}, s_{j^*} -1), (\vec{b}_{-j^*}, \increment_{j^*}(b_{j^*})]. \] Finally, if the agent of type
	$j^*$ leading the lucky permutation is $i$ itself, then the number of such
	permutations is
	\[ \begin{cases}
		(-1 + \sum_{j} s_j)! & \text{if } b_{j^*} = h \\
		0 & \text{if } b_{j^*} \neq h. 
	\end{cases}\]

	Putting everything together we get the following recurrence for
	$(\vec{s}, \vec{b})$ for the case when $\vec{b}$ is non-degenerate, $h \notin \dom(\vec{b})$, $s_{j^*} > 0$:
\begin{align} \label{eq:recurrenceassignment}
		 M[\vec{s}, \vec{b}] =
		 & \sum_{\mathclap{j \neq j^* : \atop s_j > 0 \wedge b_j \neq \textrm{nil}}} s_j \cdot M[(\vec{s}_{-j}, s_j -1), (\vec{b}_{-j}, \increment_j(b_j))] \notag \\
		 		   &+ (s_{j^*} - 1) \cdot M[(\vec{s}_{-j^*}, s_{j^*} -1), (\vec{b}_{-j^*}, \increment_{j^*}(b_{j^*})] \notag \\
		 &+ \begin{cases}
			(-1 + \sum_{j} s_j)! & \text{if } b_{j^*} = h \\
			0 & \text{if } b_{j^*} \neq h 
		\end{cases}
	\end{align}

	Once the table is filled, the probability we are after is simply
	\begin{equation*}
		RSD(R)(i)(h) = \frac{M[(d_1, \ldots, d_T), (h_1, \ldots, h_T)]}{n!},
	\end{equation*}
	where $h_j$ is the house most preferred by agents of type $j$.

	Let us bound the running time of the algorithm. First, we show how to
	efficiently compute the set of agent types present in the instance, their
	frequencies and preferences. 
	Think of the linear preference relation of some agent as a
	``string'' whose ``letters'' are houses. Notice that two agents of the same
	type give rise to the same string. In $O(|R|)$ time we can generate the set
	of all string and build a trie \citep[Ch.~5, ][]{SeWa11a} out of them, keeping a
	frequency count of how many strings of each kind we have seen, which
	correspond to how many agents of that type there are.

	Next, we need to bound the time it takes to fill the DP table. Notice that
	the total number of entries $M[\vec{s}, \vec{b}]$ can be as large as $\prod_j
	d_j \cdot m^T$ entries, since there are $\prod_j d_j$ choices for $\vec{s}$
	and $m^T$ choices for $\vec{b}$. Unfortunately, this would be too large for
	our purposes. The key observation is that not all possible vectors $\vec{s}$
	are reachable from our recurrence. In particular, notice that every time we
	apply \eqref{eq:recurrenceassignment} we decrease $\sum_j s_j$ by one and
	once all houses are assigned, the vector $\vec{b}$ must be $\textrm{nil}$
	everywhere, which means we are at the base case of the recurrence. Therefore,
	we only need to keep track of vectors $\vec{s}$ where $\sum_j (d_j - s_j)
	\leq m$. There are only ${m + T -1 \choose T -1 }$ such vectors. Thus, the
	total  number of entries we need to keep track of is bounded by ${m+T \choose
	T} \cdot m^T$. Computing $M[\vec{s}, \vec{b}]$ using
	\eqref{eq:recurrenceassignment} takes $O(T \cdot m)$ time, while using other
	cases of the recurrence takes $O(m)$ time provided the function
	$\closure(\cdot)$ is already computed, which takes $O(T\cdot m^{T+1})$ time
	overall. Adding everything up, we get that the total running time is $O( |R|
	+ { m+T \choose m} \cdot T \cdot  m^{T+1})$ as it appears in the theorem statement.
\end{proof}

\begin{corollary}\label{cor:houses}
  There is an FPT algorithm for computing $RSD(R)(i)(h)$ with parameter $m =
  \text{\# of houses}$. The running time is $O(|R| + m^{m^m})$. 
\end{corollary}

\begin{proof} 
	Notice that when we have $m$ houses, there can be at most $\sum_{k=1}^m k!$
	agent types, one for each total ordering of every subset of the $m$ houses.
	The time bound follows from plugging $T = m! (1 + \frac{2}{m}) \geq
	\sum_{k=1}^m k!$ into Theorem~\ref{th:fptassign}:

	\begin{align*}
		{ m+T \choose T } \cdot T \cdot m^{T+1} 
   		& =  { m+T \choose m } \cdot T \cdot m^{T+1}, \\
   		& \leq \frac{(m+T)^m}{m!} \cdot m! \left(1 + \frac{2}{m}\right) \cdot m^{m! \left(1 + \frac{2}{m}\right) + 1}, \\
   		& \leq \left(m+m!\left(1 + \frac{2}{m}\right)\right)^m \cdot \left(1 + \frac{2}{m}\right) \cdot m^{m! \left(1 + \frac{2}{m}\right) + 1}, \\
   		& \leq \left(m!\left(1 + \frac{3}{m}\right)\right)^m \cdot \left(1 + \frac{2}{m}\right) \cdot m^{m! \left(1 + \frac{2}{m}\right) + 1}, \\
   		& \leq m!^m \cdot \left(1 + \frac{3}{m}\right)^{m+1} \cdot m^{m! \left(1 + \frac{2}{m}\right) + 1}, \\
   		& \leq m^{m^2} \cdot O(1) \cdot m^{m! \left(1 + \frac{2}{m}\right) + 1}, \\
   		& = O(m^{m^m} ),
	\end{align*}
	where the last inequality follows from the upper bound $m! \leq e(m^{0.5}) \left(\frac{m}{e}\right)^m$, which  is an easy consequence of Stirling's approximation for factorial numbers.
\end{proof}	

\begin{corollary}\label{cor:assign-agent-types}
  There is an algorithm for computing $RSD(R)(i)(h)$ whose running time is $O(|R| + T \cdot m^{2T+1})$.
\end{corollary}

\begin{proof}
	The time bound follows directly from Theorem~\ref{th:fptassign}:
	\begin{align*}
		{ m+T \choose T } \cdot T \cdot m^{T+1} 
   		& \leq  \frac{(m+T)^T}{T!} \cdot T \cdot m^{T+1} \\
   		& \leq  (m+1)^T \cdot T \cdot m^{T+1} \\
   		& = O(T \cdot m^{2T+1}).
	\end{align*}
\end{proof}


\section{Conclusions}

\begin{table}[h]
	\centering
	\scalebox{0.95}{
	\begin{tabular}{llll}
		\toprule
		Setting & Parameter & Complexity & Reference \\
		\midrule
		Voting & $n$: \# agents & in FPT $O(|R| + m n \cdot 2^n)$
		  & Cor.~\ref{cor:numberofagents} \\
		Voting & $T$: \# agent types & in FPT $O(|R| + m T \cdot 2^{T})$
		  & Cor.~\ref{cor:numberofagenttypes} \\
		Voting & $m$: \# alternatives & in FPT $O(|R|+ m \cdot 3^m \cdot 2^{3^m})$
	    & Cor.~\ref{cor:alternatives} \\
		Voting & $q$: \# alternative types & in FPT $O(|R|+q\cdot3^q\cdot 2^{3^q})$
		  & Cor.~\ref{cor:alternativestypes} \\
		\midrule
		Assignment & $m$: \# houses & in FPT $O(|R| +m^{m^m})$
		  & Cor.~\ref{cor:houses} \\
		Assignment & $T$: \# agent types & $O(|R| + T \cdot m^{2T+1})$
		  & Cor.~\ref{cor:assign-agent-types} \\
		\bottomrule
	\end{tabular}
	}

	\caption{Parametrized time complexity bounds for \rsd}
	\label{table:summary:housing2}
\end{table}

In this paper, we presented the first parametrized complexity analysis of \rsd
both for the voting and assignment setting. For voting, we presented FPT
algorithms for parameters \# agent types and \# alternatives. For the
assignment setting, we presented an FPT algorithm for parameter \# houses.
Although an FPT algorithm for the assignment setting with parameter \# agent
types still eludes us, we showed that the problem is polynomial-time solvable if \#
agent types is  constant.
We leave as an open problem to settle the parametrized complexity of \rsd in
the assignment setting for the parameter $p = \text{\# of agent types}$.

\section*{Acknowledgments}

The authors thank the anonymous reviewers whose comments helped to improve the presentation of the paper. 
NICTA is funded by the Australian Government through the Department of Communications and the Australian Research Council through the ICT Centre of Excellence Program.

\end{document}